\documentclass[a4paper,12pt]{amsart}

\newtheorem{lemma}{Lemma}

\title[Valleys of the Supermembrane]{Measure of the potential valleys of the supermembrane theory}

\author[L.~Boulton]{Lyonell Boulton$^1$}
\address{$^1$Department of Mathematics and MIMS, Heriot-Watt University, Edinburgh, EH14 4AS, United Kingdom. \newline
Department of Mathematics, Faculty of Nuclear Sciences and Physical Engineering, Czech Technical University in Prague, Trojanova~13, 12000 Prague 2, Czech Republic.}
\email{l.boulton@hw.ac.uk; boultlyo@fjfi.cvut.cz}
\author[M.P.~Garc{\'\i}a del Moral]{Mar{\'\i}a Pilar Garc{\'\i}a del Moral$^2$}
\address{$^{2,3}$Departamento de F\'{\i}sica, Universidad de Antofagasta, Aptdo 02800, Chile.}
\email{maria.garciadelmoral@uantof.cl; alvaro.restuccia@uantof.cl}
\author[{A}. Restuccia]{{\'A}lvaro Restuccia$^3$}

\begin{document}

\maketitle

\begin{abstract}
We analyse the measure of the regularized matrix model of the supersymmetric potential valleys, $\Omega$, of the Hamiltonian of non zero modes of supermembrane theory. This is the same as  the Hamiltonian of the BFSS matrix model. We find sufficient conditions for this measure to be finite, in terms the spacetime dimension. For $SU(2)$ we show that the measure of $\Omega$ is finite for the regularized supermembrane matrix model when the transverse dimensions in the light cone gauge  $\mathrm{D}\geq 5$. This covers the important case of seven and eleven dimensional supermembrane theories, and implies the compact embedding of the Sobolev space $H^{1,2}(\Omega)$ onto $L^2(\Omega)$. The latter is a main step towards the confirmation of the existence and uniqueness of ground state solutions of the outer Dirichlet problem for the Hamiltonian of the $SU(N)$ regularized $\mathrm{D}=11$ supermembrane, and might eventually allow patching with the inner solutions.
\end{abstract}

\newpage

\section{Introduction}
The supermembrane theory was derived in \cite{bst}. Its $SU(N)$ regularization was introduced in \cite{hoppe} and in \cite{dwhn,dwmn} the $SU(N)$ regularized Hamiltonian in the light cone gauge was obtained. The zero mode eigenfunction can be described in terms of the $\mathrm{D}=11$ supergravity multiplet, however the existence of the ground state of the Hamiltonian requires a proof of the existence of a unique nontrivial eigenfunction for the nonzero modes. Moreover, in order to be identified with the 11D supergravity multiplet, it must be invariant under $SO(9)$. The existence of this ground state is still an ellusive open problem. For an (incomplete) list of contributions towards its solution, mainly in asymptotic regimes, c.f.  \cite{hoppe,hasler,fh,michishita,hl,hlt,frolich}.

In 11D supermembrane theory, the zero modes associated with the center of mass and the non-zero modes associated with the internal excitations, decouple. The groundstate of the Hamiltonian with zero eigenvalue and their associated eigenfunction can be described in terms of  the $\mathrm{D}=11$ supergravity multiplet once it is proven the existence for the non-zero modes of a unique nontrivial eigenfunction invariant under the R-symmetry $SO(9)$ \cite{dwhn}. The $SU(N)$ regularized Hamiltonian for nonzero modes coincides with the Hamiltonian of the BFSS matrix model, \cite{bfss}. This Hamiltonian was first obtained as the $0+1$ reduction of the 10D Super Yang Mills \cite{claudson,halpern}. Hence, the existence of the ground state for this matrix model, turns out to be exactly the missing step in the proof of the existence of the ground state for the $\mathrm{D}=11$ supermembrane. In \cite{sethi-stern} a prescription to compute the index of non-Fredholm operators was presented. Although a precise definition of the domain of the non-Fredholm operator ( which certainly is not the whole Hilbert space) is not given and hence the definition of the trace involved in the evaluation of the index is not precise, they introduce a prescription, which they claim, allows to compute an index of the non-Fredholm operator. Using this result  they claim to obtain a unique ground state for the Hamiltonian of the BFSS matrix model with gauge group SU(2).

In this paper we emphasise the bounds on the ground state wave function on the flat directions, which is absent in the previous works. Although we consider the SU(2) case, our approach indicates a new way to analyse the SU(N), N going to infinity, model. This is the relevant gauge group for the D=11 Supermembrane.

 Due to its complexity, a natural approach for the solution of this problem is to divide it into three parts, \cite{bgmrO,bgmrSU2,bgmrGS}. Firstly, determine the existence and uniqueness of a solution for the Dirichlet problem on a bounded region of arbitrary radius. Secondly, determine the existence and uniqueness of the solution for the Dirichlet problem on the unbounded complementary region. Thirdly, determine if both these solutions match with one another and can be smoothly patched into a single solution of the full problem. The overall state will then be the ground state of the Hamiltonian of the non-zero modes of the $\mathrm{D}=11$ supermembrane. In \cite{bgmrGS}  (also \cite{bgmrO,bgmrSU2}) we settled the first step. The  present work is about the second step.

Our proof of existence and uniqueness for bounded regions, relied on two fundamental properties of the Hamiltonian: i) its supersymmetric structure as $H=\{Q,Q^{\dag}\}$ and ii) the polynomial form of the potential expression as a function of the bosonic coordinates. We combined these two properties with iii) the Rellich-Kondrashov compact embedding theorem. Then the existence and uniqueness followed from ellipticity and the Lax-Milgram theorem for strongly coercive sesquilinear forms.

The Rellich-Kondrashov compact embedding theorem holds true for every bounded region of $\mathbb{R}^n$, but it might fail in general on unbounded regions. For the second step, one requires an estimate for the contribution of the potential to the mean value of the Hamiltonian, taking into account that this potential is unbounded from below along the sub-varieties where the bosonic potential vanishes. An estimate must therefore be obtained on the ``valleys'', denoted by $\Omega$ below, surounding these sub-varieties. In the complement of $\Omega$, the bosonic potential is the dominant part of the potential, it is strictly positive and it tends to infinity at infinity. In these ``good'' regions, the existence and uniqueness of the solution to the Dirichlet problem follows from general arguments, similar to those used for non-relativistic Schr{\"o}dinger operators. We expect that the ground state (if it exists) should extend along $\Omega$ and decay rapidly to zero in the complement of $\Omega$.

Hard work, however, has to be conducted in the interior of $\Omega$. That is, to show the existence and uniqueness of the solution to the Dirichlet problem on $\Omega$ minus a ball of finite radius. In order to achieve this goal, we devote this letter to establishing that the Rellich-Kondrashov compact embedding theorem holds true on $\Omega$ for $\mathrm{D}\geq 5$ on Sobolev spaces defined following \cite{berger-schechter}. Concretely, we show that the measure (Lebesgue measure) of the unbounded set
$$\Omega=\{x\in \mathbb{R}^n: V_{B}(x)<1\}$$
 is finite and decays at infinity for any $\mathrm{D}\geq 5$. See lemma~\ref{L2}. This includes, for the bosonic potential, the important cases of the $\mathrm{D}=7$ and $\mathrm{D}=11$ supermembrane.

 Consequently, with respect to properties i), ii) and iii), arguments analogous to i) and iii) can be made on $\Omega$. Property ii) is not valid in $\Omega$, but it might be possible to consider an estimate of the fermionic contribution to the mean value of the Hamiltonian which allows a different version of coercivity. We hope to report on this eventually.

\section{Formulation of the problem}
Before establishing our main current contribution, let us summarize the formulation of the problem. We follow the seminal work \cite{dwhn}. The $\mathrm{D}=11$ supermembrane is described in terms of the membrane coordinates $X^m$  and fermionic coordinates $\theta _{\alpha}$,  transforming as a  Majorana spinor on the target space. Both fields are scalar under worldvolume transformations. When the theory is formulated in the Light Cone Gauge the residual symmetries are global supersymmetry, the R-symmetry $SO(9)$ and a gauge symmetry, the area preserving diffeomorphisms on the base manifold.

The fields of the Hamiltonian and the wavefunction are decomposed according to the  symmetry group $SO(9)$ in such a way that the Majorana spinor is expressed in terms of the linear representations of the subgroup $SO(7)\times U(1)\subset SO(9)$.

The bosonic coordinates $X^{M}$ are decomposed as $(X^m,Z,\overline{Z})$. Where $X^m$ for $m=1,\dots,7$ are the components of a $SO(7)$ vector, and $Z,\overline{Z}$ are the complex scalars
\[Z=\frac{1}{\sqrt{2}}(X^8+iX^9)\quad \text{and} \quad \overline{Z}=\frac{1}{\sqrt{2}}(X^8-iX^9),\]
which transform under $U(1)$. The corresponding bosonic  canonical momenta is accordingly decomposed as a $SO(7)$ vector of components $P_m$ and a complex $U(1)$ momentum $\mathcal{P}$ and its conjugate $\overline{\mathcal{P}}$: $P_{M}=(P_{m},\mathcal{P},\overline{\mathcal{P}})$ where
 \[\mathcal{P}=\frac{1}{\sqrt{2}}(P^8-iP^9) \quad \text{and} \quad \overline{\mathcal{P}}=\frac{1}{\sqrt{2}}(P^8+iP^9).\]

Denoting by $\lambda_{\alpha}$ the invariant $SO(7)$ spinor of the operator associated to the fermionic coordinates. We can express it in terms of an eight component complex spinor $\theta^{\pm}$ eigenstate of $\gamma_9$,  for
$\gamma_9\theta^{\pm}=\pm\theta^{\pm}$, such that
\[\lambda^{\dag}=2^{1/4}(\theta^+-i\theta^{-})\quad \text{and} \quad \lambda=2^{1/4}(\theta^++i\theta^{-}),\]
where  $\lambda^{\dag}$ is the fermionic canonical conjugate momentum to $\lambda$.

Once the theory is regularized by means of the group $SU(N)$, the field operators are labeled by an $SU(N)$ index $A$ and they transform in the adjoint representation of the group.

The realization of the wavefunctions is formulated in terms of the $2^{8(N^2-1)}$ an irreducible representation of the Clifford algebra span by $(\lambda^{\dag}+\lambda)$ and $i(\lambda^{\dag}-\lambda)$ in the fermion Fock space. The Hilbert space of physical states consists of the wavefunctions which takes values in the fermion Fock space.

Once it is shown that the zero mode states transform under $SO(9)$ as a $[(44\oplus 84)_{\mathrm{bos}}\oplus 128_{\mathrm{fer}}]$ representation which corresponds to the massless $\mathrm{D}=11$ supergravity supermultiplet,
the construction of the ground state wave function reduces to finding a nontrivial solution to
\begin{equation*}\label{A}H\Psi=0\end{equation*}
where $H=\frac{1}{2}M$ and $\Psi\equiv \Psi^{\mathrm{non-zero}}$. The latter  is required to be a singlet under $SO(9)$ and $M$ is the mass operator of the supermembrane.
The Hamiltonian associated to the the regularized mass operator of the supermembrane \cite{dwhn} is
\[
\begin{aligned}
H&=\frac{1}{2}M^2=-\Delta+V_B+V_F\\
\Delta&=\frac{1}{2}\frac{\partial^2}{\partial X^i_A\partial X_i^A}+\frac{1}{2}\frac{\partial^2}{\partial Z_A\partial \overline{Z}^A}\\
V_{B}&=\frac{1}{4}f_{AB}^Ef_{CDE}\{X_i^AX_j^BX^{iC}X^{jD}+4X_i^AZ^BX^{iC}\overline{Z}^{D}+2Z^A\overline{Z}^B\overline{Z}^{C}Z^{D}\}\\
V_F&=if_{ABC}X_i^A\lambda_{\alpha}^B\Gamma_{\alpha\beta}^i\frac{\partial}{\partial\lambda_{\beta C}}
+\frac{1}{\sqrt{2}}f_{ABC}(Z^A\lambda_{\alpha}^B\lambda_{\alpha}^C-\overline{Z}^A\frac{\partial}{\partial \lambda_{\alpha\beta}}\frac{\partial}{\partial\lambda_{\alpha C}}).
\end{aligned}
\]
The generators of the local $SU(N)$ symmetry are
$$\varphi^A =f^{ABC}\left( X_{i}^B\partial_{X_i^C}+Z_B\partial_{Z^C}+\overline{Z}_B\partial_{\overline{Z}^C}
+\lambda_{\alpha}^{B}\partial_{\lambda_{\alpha}^C}\right).$$

From the supersymmetric algebra, it follows that the Hamiltonian can be express in terms of the supercharges as
$$H =\{Q_{\alpha},Q^{\dagger}_{\alpha}\}$$
for the physical subspace of solutions, given by the kernel of the first class constraint $\varphi^A$ of the theory, that is
\begin{equation*} \label{constraint} \varphi^A\Psi=0.\end{equation*}

The supercharges associated to  modes invariant under $SO(7)\times U(1)$ are given explicitly in \cite{dwhn} as
\[
\begin{aligned}Q_{\alpha}&=\left\{-i\Gamma_{\alpha\beta}^i\partial_{X_i^A}+\frac{1}{2}f_{ABC}X_i^B X_j^C \Gamma^{ij}_{\alpha\beta}-f_{ABC} Z^B\overline{Z}^C\delta_{\alpha\beta}\right\}\lambda_{\beta}^A\\
&+\sqrt{2}\left\{\delta_{\alpha\beta}\partial_{Z^A}+i f_{ABC}X_i^B \overline{Z}^C \Gamma^i_{\alpha\beta} \right\}\partial_{\lambda_{\beta}^A}\end{aligned}\]
and
\[
\begin{aligned}Q_{\alpha}^{\dagger}&=\left\{i\Gamma_{\alpha\beta}^i\partial_{X_i^A}+\frac{1}{2}f_{ABC}X_i^B X_j^C \Gamma^{ij}_{\alpha\beta}+f_{ABC} Z^B\overline{Z}^C\delta_{\alpha\beta}\right\}\partial_{\lambda_{\beta}^A}\\
&+\sqrt{2}\left\{-\delta_{\alpha\beta}\partial_{Z^A}+i f_{ABC}X_i^B \overline{Z}^C \Gamma^i_{\alpha\beta} \right\}\lambda_{\beta}^A.\end{aligned}\]
The corresponding superalgebra satisfies, \cite{dwhn}
\[\begin{aligned}&\{Q_{\alpha},Q_{\beta}\}=2\sqrt{2}\delta_{\alpha\beta}\overline{Z}^A\varphi_A,\\
&\{Q_{\alpha}^{\dagger},Q_{\beta}^{\dagger}\}=2\sqrt{2}\delta_{\alpha\beta}Z^A\varphi^A,\\
&\{Q_{\alpha},Q^{\dagger}_{\beta}\}=2\delta_{\alpha\beta}H-2i\Gamma^i_{\alpha\beta}\varphi_A.\end{aligned}\]
These must annihilate the physical states. The Hamiltonian $H$ is a positive operator which annihilates $\Psi$, on the physical subspace, if and only if  $\Psi$ is a \emph{singlet} under supersymmetry\footnote{$\Psi_0$, the zero mode wave function, in distinction is a supermultiplet under supersymmetry.}. In such a case,
$$Q_{\alpha}\Psi=0\quad \text{and}\quad Q_{\alpha}^{\dagger}\Psi=0.$$

The latter ensures that the wavefunction is massless, however it does not guarantee that the ground-state wave function is the corresponding supermultiplet associated to supergravity. For this, $\Psi$ must also become a singlet under $SO(9)$. The spectrum of $H$ in $L^2(R^n)$ is continuous \cite{dwln}, comprising the segment $[0,\infty)$.

The previous supersymmetric structure implies the following. This is the property i) described above. The current interest is the case $\Sigma=\Omega$.

\begin{lemma} \label{lemma2}
Let $\Sigma\subset \mathbb{R}^n$ be a region with smooth boundary.
If $u\in H_0^1(\Sigma)$ satisfies $Qu=Q^{\dag}u=0$ in $\Sigma$ then $u=0$ in $\overline{\Sigma}$.
\end{lemma}
\begin{proof}
The argument is the same as in \cite{bgmrGS}. If $u$ satisfies $Qu=Q^{\dag}u=0$, then $u$ is analytic in $\Sigma$, since the potential is analytic in $x$. Hence $Qu=Q^{\dag}u=0$ also on the boundary $\partial\Sigma$. Then using the explicit expression of $Q$ and $Q^{\dag}$, we obtain that the normal derivative of $u$ on $\partial \Sigma$ is also zero. We thus have $u=0$ and $\partial_n u=0$ simultaneously on $\partial\Sigma$. By virtue of the Cauchy-Kowaleski theorem on $\partial\Sigma$, $u=0$ in a neighborhood of $\partial\Sigma$. Since $u$ is analytic we conclude that $u=0$ in $\overline{\Sigma}$.
\end{proof}

\section{Analysis of the Lebesgue measure of the bosonic valleys}

We simplify the proof of our main result below by denoting the bosonic coordinates with $X^A_i$, for $i=1,\dots,\mathrm{D}$ and $A=1,2,3$ the $SU(2)$ index. We will denote a vector of $3\times \mathrm{D}$ components by means of D vectors of 3 components: $\vec{X}_i\in \mathbb{R}^3$, $i=1,\dots,\mathrm{D}.$ We denote with a single bar, $\vert\cdot\vert$, the  Euclidean norm on any number of components. The bosonic potential reduces to
\[
V_B=\frac{1}{2}\sum_{i,j=1}^D\vert \vec{X}_i\wedge \vec{X}_j\vert^2.
\]
Below we repeatedly use the following property without further mention. If $R$ is any rotation of $\mathbb{R}^3$, then
\[
     V_B(\vec{X}_1,\ldots,\vec{X}_{\mathrm{D}})=
    V_B(R\vec{X}_1,\ldots,R\vec{X}_{\mathrm{D}}).
\]

For $a_0\geq 0$, let
\[
    \Omega_{a_0}=\{ (\vec{X}_1,\ldots,\vec{X}_{\mathrm{D}}): V_B<1, \vert (\vec{X}_1,\ldots,\vec{X}_{\mathrm{D}})\vert \geq a_0\}.
\]
So that $\Omega=\Omega_0$. We denote the (Lebesgue) measure of any of these sets by $\mu(\Omega_{a_0})$.

\begin{lemma}\label{L2}
For $D\geq 5$, $\mu(\Omega)$ is finite and
\[\lim_{a_0\to\infty} \mu(\Omega_{a_0})= 0.\]
\end{lemma}
\begin{proof}
We firstly notice that we can exchange the orders of integration below, because $V_B$ is a polynomial in its components. Fix one direction $\vec{e}$ and consider the change of variables that rotates $\vec{X}_1$ to $a\vec{e}$. Then
\[
       \mu(\Omega_{a})=4\pi \int_{a_0}^\infty a^2 \mathrm{d}a
        \int_{\hat{\Omega}_{a_0}} \prod_{j=2}^{\mathrm{D}}\mathrm{d}
        x_{j}^1 \mathrm{d}
        x_{j}^2 \mathrm{d}
        x_{j}^3
\]
where
\[
    \hat{\Omega}_{a}=\{(a\vec{e},\vec{X}_2,\ldots,\vec{X}_{\mathrm{D}}): V_B<1, \vert (a\vec{e},\vec{X}_2,\ldots,\vec{X}_{\mathrm{D}})\vert \geq a_0\}.
\]
Considering
\[\vec{X}_1= a\vec{e},\quad \vec{X}_i=b_i\vec{e}+X_i^{\perp}
\]
where $e\cdot X_i^{\perp}=0,\, i=2,\dots,\mathrm{D}$, we have
 \[
\vert \vec{X}_i\wedge\vec{X}_j\vert^2 =\vert b_i{X}_j^{\perp}- b_j{X}_i^{\perp}\vert^2 +\vert {X}_i^{\perp}\wedge {X}_j^{\perp}\vert^2.
\]
Write now, $X_i^{\perp}=c_i\vec{e}_2+d_i\vec{e_3}$ with $\vec{e}_2\cdot\vec{e}_3=0$ and $\vert\vec{e}_2\vert=\vert\vec{e}_2\vert=1$. Substituting in $V_B$ we get
\begin{equation} \label{potBCD}
\begin{aligned}V_B= a^2 (\vert C\vert^2+& \vert D\vert^2)+ \\ & \vert B\vert^2\vert C\vert^2-(B\cdot C)^2+ \\ & \vert B\vert^2\vert D\vert^2-(B\cdot D)^2+ \\ & \vert C\vert^2\vert D\vert^2-(C\cdot D)^2 \end{aligned}
\end{equation}
where we have denoted by $B;C;D$ the points in $\mathbb{R}^{D-1}$ with components
\[
B=\begin{pmatrix}
b_2\\b_3\\ \vdots \\ b_D
\end{pmatrix};C=\begin{pmatrix}
c_2\\c_3\\ \vdots \\ c_D
\end{pmatrix};D=\begin{pmatrix}
d_2\\d_3\\ \vdots \\ d_D
\end{pmatrix}.
\]
Then
\begin{equation} \label{lebeme}
\mu(\Omega_{a_0})=4\pi\int_{a_0}^\infty a^2 \mathrm{d}a \int_{\tilde{\Omega}_a} \prod_{i=2}^{D}\mathrm{d}b_i\prod_{j=2}^{D}\mathrm{d}c_j\prod_{k=2}^{D}\mathrm{d}d_k
\end{equation}
where
\[
     \tilde{\Omega}_a=\{(B;C;D):\text{RHS of }\eqref{potBCD}<1,\,
    a^2+|B|^2+|C|^2+|D|^2\geq a_0\}.
\]

In order to estimate the integrals in \eqref{lebeme} we change variables
to
\[
    C=\alpha \frac{B}{|B|} + C_{\perp} \qquad
    D=\beta \frac{B}{|B|} + B_{\perp}
\]
where $B\cdot C_{\perp}=B\cdot D_{\perp}=0$ so that $\alpha=C\cdot \frac{B}{|B|}$ and $\beta = D\cdot \frac{B}{|B|}$. The potential becomes
\begin{align*}
    V_B&=a^2(|C|^2+|D|^2)+|\beta|^2(|C_{\perp}|^2+|D_\perp|^2) \\
&\geq (a^2+|B|^2)(|C_\perp|^2+|D_{\perp}|^2)+a^2(\alpha^2+\beta^2).
\end{align*}
For $a$ and $B$ fixed, the region
\[
    E_{a,B}=\{(\alpha,C_{\perp},\beta,D_{\perp}):(a^2+|B|^2)(|C_\perp|^2+|D_{\perp}|^2)+a^2(\alpha^2+\beta^2)<1\}
\]
is an ellipsoid which contains
\[
    \{(\alpha,C_{\perp},\beta,D_{\perp}):V_B<1\}.
\]
Then
\begin{align*}
    \mu(\Omega_{a_0})& \leq k_1(\mathrm{D})\int_{a_0}^\infty
     a^2 \mathrm{d}a \int_{B\in \mathbb{R}^{\mathrm{D}-1}}
     \mu(E_{a,B}) \prod_{i=2}^{D}\mathrm{d}b_i \\
&= k_1(\mathrm{D})
\int_{a_0}^\infty
     \mathrm{d}a \int_{B\in \mathbb{R}^{\mathrm{D}-1}}
     \frac{\prod_{i=2}^{D}\mathrm{d}b_i}{(a^2+|B|^2)^{\mathrm{D}-2}} \\
   &= k_2(\mathrm{D}) \int_{a_0}^\infty \int_{0}^{\infty}
     \frac{u^{\mathrm{D-2}}\mathrm{d}u\mathrm{d}a}{(a^2+u^2)^{\mathrm{D}-2}} \\   &= k_3(\mathrm{D}) \int_{a_0}^\infty \frac{\mathrm{d}a}{a^{\mathrm{D}-3}}
\end{align*}
where $k_l(\mathrm{D})$ are constants.
Finaly notice that the right hand side is finite for all $a_0> 0$ and decreases to $0$ as $a_0\to\infty$, whenever $\mathrm{D}\geq 5$.
\end{proof}

We denote by $H^p (\Omega)$ and $H^p_0(\Omega)$, respectively, the Sobolev spaces $H^{p,2}(\Omega)$ and $\mathring{W}^{p,2}(\Omega)$ in the notation of \cite{berger-schechter}. A crucial observation here is the fact that these spaces are amenable to patching inner and outer domains in the solution of Dirichlet problems. We recall that $H^p(\Omega)$ is the Hilbert space arising from restricting to $\Omega$ functions in the Sobolev space $H^p(\mathbb{R}^n)$, the norm being the infimum of the Sobolev norm over all possible extensions. We also recall that $H^p_0(\Omega)$ is the completion with respect to this norm of
all smooth functions with support a compact subset of $\Omega$.  By combining lemma~\ref{L2} with Theorem~2.8 of \cite{berger-schechter}, we immediately obtain the remarkable property that that both $H^p (\Omega)$ and $H^p_0(\Omega)$ are compactly embedded into $L^2(\Omega)$ for $\mathrm{D}\geq 5$.

\section{Bounds for the fermionic potential} In order to prove the existence and uniqueness of the solution to the outer Dirichlet problem we need a bound on the contribution of the mean value of the fermionic potential $(u,V_Fu)_{L^2(\Omega)}$. We notice that the fermionic potential is linear on the bosonic coordinates. Then

$$\vert(u,V_F u)\vert_{L^2(\Omega)}\le C(u,\rho u)_{L^2 (\Omega)}$$ for some $C>0$,
where $\rho^2 =\vert x\vert^2\equiv a^2+\vert B\vert^2+\vert C\vert^2+\vert D\vert^2 $. From Lemma~\ref{L2}, it is possible to derive the following bound,

$$\vert(u,V_F u)\vert_{L^2(\Omega)}\le C \left\|u\right\|^2_{L^{\infty}(\Omega)}\int_{\Omega}\rho$$
provided that
\begin{equation} \label{intcond}
\int_{\Omega}\rho<\infty.
\end{equation}  In turns, we have the following results which follows from similar arguments as those of lemma~\ref{L2}.

\begin{lemma}
If $D>5$, then  $\int_{\Omega}\rho^2<\infty$.
\end{lemma}

As a corollary of this lemma, we obtain that for $D>5$, also \eqref{intcond} holds true. We hope to discuss a sharp bound of the form $(u,V_Fu)_{L^2 (\Omega)}\le k\left\|u\right\|^2 _{H^1(\Omega)}$ in future work.
	
\section{Conclusions} We have shown that the volume of the valleys, the set $\Omega$, is finite when the dimension of the target space on which the supermembrane theory is formulated is greater than or equal to five (transverse) dimensions. This include the important 7 and 11-dimensional supermembranes.
Using a framework due to Berger and Schechter we have shown that on $\Omega$, the embeddings of $H^1(\Omega)$ and $H^1_0(\Omega)$ onto $L^2(\Omega)$ are compact. Notice that this property is not related to the zero point energy of the bosonic membrane, a main observation in the argument to conclude that the bosonic membrane has discrete spectrum.  In fact this result does not depend on the target space dimension. Furthermore, we have argued with supporting evidence about bounds for the mean value of the fermionic potential.  We have then shown  properties i) and iii) proposed in the introduction and claim that it is possible to get an appropiate bound for the mean value of the fermionic potential on any wave function in $H_0^1(\Omega)$.
The complete proof of the three statements will allow determining the existence and uniqueness of the solution of the outer Dirichlet problem on  the valleys of the bosonic potential.

\section{Acknowledgements}  AR and MPGM were partially supported by Projects Fondecyt 1161192 (Chile). LB kindly acknowledges support from MINEDUC-UA project code ANT1755. This research was initiated as part of a visit of this author to the Universidad de Antofagasta in April~2018 and completed when he was on a study leave at the {\v C}esk{\'e} Vysok{\'e} U{\v c}en{\'i} Technick{\'e} v Praze in November~2018.


\begin{thebibliography}{10}

\bibitem{bfss}
T.~Banks, W.~Fischler, S.~H. Shenker, and L.~Susskind.
\newblock M theory as a matrix model: a conjecture.
\newblock {\em Phys. Rev. D (3)}, 55(8):5112--5128, 1997.

\bibitem{berger-schechter}
M.~S. Berger and M.~Schechter.
\newblock Embedding theorems and quasi-linear elliptic boundary value problems
  for unbounded domains.
\newblock {\em Trans. Amer. Math. Soc.}, 172:261--278, 1972.

\bibitem{bst}
E.~Bergshoeff, E.~Sezgin, and P.~K. Townsend.
\newblock Supermembranes and eleven-dimensional supergravity.
\newblock {\em Phys. Lett. B}, 189(1-2):75--78, 1987.

\bibitem{bgmrSU2}
L.~Boulton, M.~P. Garcia~del Moral, and A.~Restuccia.
\newblock Massless ground state for a compact {$SU(2)$} matrix model in 4{D}.
\newblock {\em Nuclear Phys. B}, 898, 2015.

\bibitem{bgmrO}
L.~Boulton, M.~P. Garcia~del Moral, and A.~Restuccia.
\newblock On the groundstate of octonionic matrix models in a ball.
\newblock {\em Physics Letters B}, 744:260--262, 2015.

\bibitem{bgmrGS}
L.~Boulton, M.~P. Garcia~del Moral, and A.~Restuccia.
\newblock The ground state of the {$D=11$} supermembrane and matrix models on
  compact regions.
\newblock {\em Nuclear Phys. B}, 910:665--684, 2016.

\bibitem{claudson}
M.~Claudson and M.~B. Halpern.
\newblock Supersymmetric ground state wave functions.
\newblock {\em Nuclear Phys. B}, 250(4):689--715, 1985.

\bibitem{sethi-stern}
S.~Sethi and M.~Stern 
\newblock D-brane Bound States Redux.
\newblock {\em Commun.Math.Phys.} ,194 (1998) 675-705 

\bibitem{dwhn}
B.~de~Wit, J.~Hoppe, and H.~Nicolai.
\newblock On the quantum mechanics of supermembranes.
\newblock {\em Nuclear Phys. B}, 305(4, FS23):545--581, 1988.

\bibitem{dwln}
B.~de~Wit, M.~L\"{u}scher, and H.~Nicolai.
\newblock The supermembrane is unstable.
\newblock {\em Nuclear Phys. B}, 320(1):135--159, 1989.

\bibitem{dwmn}
B.~de~Wit, U.~Marquard, and H.~Nicolai.
\newblock Area-preserving diffeomorphisms and supermembrane lorentz invariance.
\newblock {\em Comm. Math. Phys.}, 128(1):39--62, 1990.

\bibitem{frolich}
J.~Fr\"{o}hlich, G.~M. Graf, D.~Hasler, J.~Hoppe, and S.-T. Yau.
\newblock Asymptotic form of zero energy wave functions in supersymmetric
  matrix models.
\newblock {\em Nuclear Phys. B}, 567(1-2):231--248, 2000.

\bibitem{fh}
J.~Fr\"{o}hlich and J.~Hoppe.
\newblock On zero-mass ground states in super-membrane {M}atrix models.
\newblock {\em Comm. Math. Phys.}, 191(3):613--626, 1998.

\bibitem{halpern}
M.~B. Halpern and C.~Schwartz.
\newblock Asymptotic search for ground states of {${\rm SU}(2)$} {M}atrix
  theory.
\newblock {\em Internat. J. Modern Phys. A}, 13(25):4367--4408, 1998.

\bibitem{hasler}
D.~{Hasler} and J.~{Hoppe}.
\newblock {Asymptotic Factorisation of the Ground-State for {SU(N)}-invariant
  Supersymmetric Matrix-Models}.
\newblock {\em ArXiv High Energy Physics - Theory e-prints hep-th/0206.043},
  2002.

\bibitem{hoppe}
J.~{Hoppe}.
\newblock {Asymptotic Zero Energy States for {SU(N} greater or equal 3)}.
\newblock {\em ArXiv High Energy Physics - Theory e-prints hep-th/0991.2163},
  1999.

\bibitem{hl}
J.~{Hoppe} and D.~{Lundholm}.
\newblock {On the Construction of Zero Energy States in Supersymmetric Matrix
  Models {IV}}.
\newblock {\em ArXiv High Energy Physics - Theory e-prints hep-th/0706.0353},
  2007.

\bibitem{hlt}
J.~Hoppe, D.~Lundholm, and M.~Trzetrzelewski.
\newblock Construction of the zero-energy state of {${\rm SU}(2)$}-matrix
  theory: near the origin.
\newblock {\em Nuclear Phys. B}, 817(3):155--166, 2009.

\bibitem{michishita}
Y.~Michishita and M.~Trzetrzelewski.
\newblock Towards the ground state of the supermembrane.
\newblock {\em Nuclear Phys. B}, 868(2):539--553, 2013.

\end{thebibliography}
\end{document}